\documentclass[a4paper,UKenglish]{lipics-v2016}
\usepackage{microtype}
\title{On Natural Deduction for Herbrand Constructive Logics II: Curry-Howard Correspondence for Markov's Principle in First-Order Logic and Arithmetic}
\titlerunning{Curry-Howard for Markov's Principle in First-Order Logic and Arithmetic}

\author{Federico Aschieri\footnote{Funded by the Austrian Science Fund FWF Lise Meitner grant M 1930--N35}}
\author{Matteo Manighetti\footnote{Funded by the Vienna Science Fund WWTF project VRG12-004}}
\affil{Institut f\"ur Diskrete Mathematik und Geometrie\\ Technische Universit\"at Wien\\ Wiedner Hauptstra\ss e 8-10/104, 1040, Vienna, Austria}
\authorrunning{F. Aschieri and M. Manighetti} 

\Copyright{John Q. Open and Joan R. Access}

\subjclass{F.4.1 Proof Theory}
\keywords{Markov's Principle, first-order logic, natural deduction, Curry-Howard}

\EventEditors{}
\EventNoEds{0}
\EventLongTitle{}
\EventShortTitle{}
\EventAcronym{}
\EventYear{}
\EventDate{}
\EventLocation{}
\EventLogo{}
\SeriesVolume{0}
\ArticleNo{1}

\usepackage{bussproofs}
\usepackage{amsthm}
\usepackage{amssymb}
\usepackage[inline]{enumitem}
\usepackage{cleveref}
\usepackage{upgreek}

\newcommand{\pair}[2]{\langle #1, #2 \rangle}
\newcommand{\proj}[2]{#2 \pi_{#1}}
\newcommand{\inj}[2]{\upiota_{#1}(#2)}
\newcommand{\case}[3]{#1[#2, #3]}
\newcommand{\enc}[2]{(#1, #2)}
\newcommand{\dest}[3]{#1[(\alpha,#2).#3]}
\newcommand{\Nat}                      { {\tt N} }
\newcommand{\marp}{\mathsf{MP}}
\newcommand{\HMP}{\mathsf{HMP}}
\newcommand{\IL}{\mathsf{IL}}
\newcommand{\HA}{\mathsf{HA}}
\newcommand{\EM}{\mathsf{EM}}
\newcommand{\MP}[1]{\mathcal{M}_{#1}}
\newcommand{\Ez}[3]{{ #2\, |_{#1}\, #3}}
\newcommand{\E}[3]{{ #2 \parallel_{#1} #3}}
\newcommand{\prop}[1]    {{\mathsf{#1}}} 
\newcommand{\emeno}[1]{\mathsf{EM}_{#1}^-}
\newcommand{\emme}[1]{\mathsf{EM}_{#1}}
\newcommand{\Hyp}[3]{{\mathtt{H}_{#1}^{\forall {#2} {#3}}}}
\newcommand{\Hypz}[1]{{\mathtt{H}^{#1}}}
\newcommand{\Wit}[3]{\mathtt{W}_{#1}^{\exists {#2} {#3} }}

\newcommand{\rec}{\mathsf{R}}
\newcommand{\suc}{\mathsf{S}}
\newcommand{\True}                     { {\texttt{True}} }
\newcommand{\False}                     { {\texttt{False}} }
\theoremstyle{plain}
\newtheorem{proposition}[theorem]{Proposition}

\crefformat{section}{Section \S#2#1#3}

\begin{document}
\maketitle
 \begin{abstract}
 Intuitionistic first-order logic extended with a restricted form of Markov's principle is constructive and admits a Curry-Howard correspondence, as shown by Herbelin. We provide a simpler proof of that result and then we study intuitionistic first-order logic extended with unrestricted Markov's principle. Starting from classical natural deduction, we restrict  the excluded middle and we obtain a natural deduction system and a parallel Curry-Howard isomorphism for the logic. We show that proof terms for existentially quantified formulas reduce to a list of individual terms representing all possible witnesses. As corollary, we derive that the logic is Herbrand constructive: whenever it proves any existential formula, it proves also an Herbrand disjunction for the formula. Finally, using the techniques just introduced, we also provide a new computational interpretation of Arithmetic with Markov's principle.
\end{abstract}

\section{Introduction}

 Markov's Principle was introduced by Markov in the context of his theory of Constructive Recursive Mathematics (see \cite{Troelstra2}). Its original formulation is tied to Arithmetic: it states that given a recursive function $f: \mathbb{N} \rightarrow \mathbb{N}$, if it is impossible that for every natural number $n$, $f(n)\neq 0$, then there exists a $n$ such that $f(n)=0$.  Markov's original argument for justifying it was simply the following: if it is not possible that for all $n$, $f(n)\neq 0$, then by computing in sequence $f(0), f(1), f(2), \ldots$, one will eventually hit a number $n$ such that $f(n)=0$ and will \emph{effectively} recognize it as a witness. 
 
 Markov's principle is readily formalized in Heyting Arithmetic as the axiom scheme
$$\lnot\lnot \exists \alpha^{\Nat} {P}\rightarrow \exists \alpha^{\Nat} {P}$$ where ${P}$ is a primitive recursive predicate \cite{Troelstra}.  When added to Heyting Arithmetic, Markov's principle gives rise to a \emph{constructive} system, that is, one enjoying the disjunction and the existential witness property \cite{Troelstra} (if a disjunction is derivable, one of the disjoints is derivable too, and if an existential statement is derivable, so it is one instance of it). Furthermore, witnesses for any provable existential formula can be effectively computed using either Markov's unbounded search and Kleene's realizability \cite{Kleene} or much more efficient functional interpretations \cite{Godeldialectica, AZMarkov}.

\subsection{Markov's Principle in First-Order Logic}

The very shape of Markov's principle makes it also a purely logical principle, namely an instance of the double negation elimination axiom. But in pure logic, what exactly should Markov's principle correspond to? In particular, what class of formulas should $P$ be restricted to? Since Markov's principle was originally understood as a constructive principle, it is natural to restrict $P$ as little as possible, while maintaining the logical system as constructive as possible. As proven by Herbelin \cite{Herbelin}, it turns out that asking that $P$ is propositional and with no implication $\rightarrow$ symbols guarantees that intuitionistic logic extended with such a version of Markov's principle is constructive. The proof of this result employs a Curry-Howard isomorphism based on a mechanism for raising and catching exceptions. As opposed to the aforementioned functional interpretations of Markov's principle, Herbelin's calculus  is fully isomorphic to an intuitionistic logic: there is a perfect match between reduction steps at the level of programs and detour eliminations at the level of proofs. Moreover, witnesses for provable existential statements are computed by the associated proof terms. Nevertheless, as we shall later show, the mechanism of throwing exceptions plays no role during these computations: intuitionistic reductions are entirely enough for computing witnesses. 

A question is now naturally raised: as no special mechanism is required for witness computation using Herbelin's restriction of Markov's principle, can the first be further relaxed so that the second becomes stronger as well as computationally \emph{and} constructively meaningful? Allowing the propositional matrix $P$ to contain implication destroys the constructivity of the logic. It turns out, however, that \emph{Herbrand constructivity} is preserved.  An intermediate logic is called {Herbrand constructive} if it enjoys a strong form of Herbrand's Theorem \cite{Buss, AschieriZH}: for \emph{every} provable formula $\exists \alpha\, A$, the logic proves as well an Herbrand disjunction
 $$A[m_{1}/\alpha]\lor \ldots \lor A[m_{k}/\alpha]$$
 So the Markov principle we shall interpret in this paper is 
 $$\marp: \lnot\lnot \exists \alpha\, {P}\rightarrow \exists \alpha\, {P} \mbox{\qquad($P$ propositional formula)}$$
 and show that when added to intuitionistic first-order logic, the resulting system is Herbrand constructive. This is the most general form of Markov's principle that allows a significant constructive interpretation: we shall show how to non-trivially compute lists of witnesses for provable existential formulas  thanks to an exception raising construct and a parallel computation operator. $\marp$ can also be used in conjunction with negative translations to compute Herbrand disjunctions in classical logic, something which is not possible with Herbelin's form of Markov's principle.



\subsection{Restricted Excluded Middle}

The Curry-Howard correspondence we present here is by no means an ad hoc construction, only tailored for Markov's principle. It is a simple restriction of the Curry-Howard correspondence for classical first-order logic introduced in \cite{AschieriZH}, where classical reasoning is formalized by the excluded middle inference rule:
\begin{prooftree}
  \AxiomC{$\Gamma, a: \forall x \, \prop{Q} \vdash u: C$}
  \AxiomC{$\Gamma, a:  \exists x\, \neg \prop{Q} \vdash v: C$}
  \RightLabel{$\mathsf{EM}$}
  \BinaryInfC{$\Gamma\vdash  \E{a}{u}{v} : C$}
\end{prooftree}
It is enough to restrict the conclusion $C$ of this rule to be a simply existential statement and the $\prop{Q}$ in the premises $\forall x \, \prop{Q},  \exists x\, \neg \prop{Q}$ to be propositional.  We shall show that the rule is intuitionistically equivalent to  $\marp$. With our approach, strong normalization is just inherited and the transition from classical logic to intuitionistic logic with $\marp$ is smooth and natural. 

 \subsection{Markov's Principle in Arithmetic}
 
 We shall also provide a computational interpretation of Heyting Arithmetic with $\marp$. The system is constructive and witnesses for provable existential statements  can be computed. This time, we shall restrict the excluded middle as formalized in $\cite{ABB}$ and we shall directly obtain the desired Curry-Howard correspondence. As a matter of fact, the interpretation of $\marp$ in Arithmetic ends up to be a simplification of the methods we use in first-order logic, because the decidability of atomic formulas greatly reduces parallelism and eliminates case distinction on the truth of atomic formulas. 
 
 \subsection{Plan of the Paper}
 In \Cref{sec:herb}, we provide a simple computational interpretation of first-order intuitionistic logic extended with Herbelin's restriction of Markov's principle. We also show that the full Markov principle $\marp$ cannot be proved in that system. In \Cref{sec:ilemeno}, we provide a Curry-Howard correspondence for intuitionistic logic with $\marp$, by restricting the excluded middle, and show that the system is Herbrand constructive. In \Cref{sec:arithmetic}, we extend the Curry-Howard to Arithmetic with $\marp$ and show that the system becomes again constructive.

\section{Herbelin's Restriction of Markov's Principle}
\label{sec:herb}

In \cite{Herbelin}  Herbelin introduced a Curry-Howard isomorphism for an extended intuitionistic logic. By employing exception raising operators and new reduction rules, he proved that the logic is constructive and can derive the axiom scheme
$$\HMP: \neg \neg \exists \alpha \, P \to \exists \alpha \, P \text{\qquad($P$ propositional and $\rightarrow$  not occurring in $P$)}$$
Actually, Herbelin allowed $P$ also to contain existential quantifiers, but in that case the axiom scheme is intuitionistically equivalent to $  \neg \neg \exists \alpha_{1}\ldots \exists \alpha_{n} \, P \to \exists\alpha_{1}\ldots \exists \alpha_{n} \, P$, again with $P$ propositional and $\rightarrow$  not occurring in $P$. All of the methods of our paper apply to this case as well, but for avoiding trivial details, we keep the present $\HMP$.  

Our first goal is to show that $\HMP$
has a simpler computational interpretation and to provide a straightforward proof that, when added on top of first-order intuitionistic logic, $\HMP$ gives rise to a constructive system. In particular, we show that the ordinary Prawitz reduction rules for intuitionistic logic and thus the standard Curry-Howard isomorphism \cite{Sorensen}  are enough for extracting witnesses for provable existential formulas. The crucial insight, as we shall see, is that $\HMP$ can never  actually appear  in the head of a closed proof term having existential type. It thus plays no computational role in computing witnesses; it plays rather a logical role, in that it may be used to prove the correctness of the witnesses.

To achieve our goal, we consider the usual natural deduction system for intuitionistic first-order logic \cite{Prawitz, Sorensen}, to which we add $\HMP$.  Accordingly, we add to the associated lambda calculus the constants $\MP{P}: \neg \neg \exists \alpha \, P \to \exists \alpha \, P$. The resulting Curry-Howard system is called $\IL+\HMP$ and is presented in \cref{fig:system}. 
The reduction rules for $\IL+\HMP$ presented in \cref{fig:red-ilmp} are just the ordinary ones of lambda calculus. On the other hand, $\MP{P}$ has no computational content and thus no associated reduction rule. Of course, the strong normalization of $\IL+\HMP$  holds by virtue of the result for standard intuitionistic Curry-Howard.
\begin{theorem}
  The system $\IL+\HMP$ is strongly normalizing
\end{theorem}

 \begin{figure*}[!htb]
 
\footnotesize{

\begin{description}

\item[Grammar of Untyped Proof Terms]
\[t,u, v::=\ x\  |\ tu\ |\ tm\ |\ \lambda x\, u\  |\ \lambda \alpha\, u\ |\ \pair{t}{u}\ |\ \proj{0}{u}\ |\ \proj{1}{u} \ |\ \inj{0}{u}\ |\ \inj{1}{u}\  |\ \case{t}{x.u}{y.v}\ |\ (m,t)\ |\ t[(\alpha, x). u]\ |\ \Hypz{\prop{\bot\rightarrow P}} |\ \MP{P}\] 
where $m$ ranges over terms of the first-order language of formulas $\mathcal{L}$, $x$ over proof-term variables, $\alpha$ over first-order variables.

\item[Contexts] With $\Gamma$ we denote contexts of the form $x_1:A_1, \ldots, x_n:A_n$, where each $x_{i}$ is a proof-term variable, and $x_{i}\neq x_{j}$ for $i\neq j$.

\item[Axioms] 
$\begin{array}{c}   \Gamma, x:{A}\vdash x: A
\end{array}\ \ \ \ $
$\begin{array}{c}   \Gamma \vdash  \MP{P}: \neg \neg \exists \alpha \, \prop{P} \to \exists \alpha \, \prop{P}
\end{array}\ \ \ \ $
$\begin{array}{c}   \Gamma \vdash  \Hypz{\prop{\bot\rightarrow P}}:   \bot\rightarrow \prop{P}
\end{array}\ \ \ \ $ \\

\item[Conjunction] 
$\begin{array}{c}  \Gamma \vdash u:  A\ \ \ \Gamma\vdash t: B\\ \hline \Gamma\vdash \pair{u}{t}:
A\wedge B
\end{array}\ \ \ \ $
$\begin{array}{c} \Gamma \vdash u: A\wedge B\\ \hline\Gamma \vdash \proj{0}{u}: A
\end{array}\ \ \ \ $
$\begin{array}{c}  \Gamma \vdash u: A\wedge B\\ \hline \Gamma\vdash \proj{1}{u} : B
\end{array}$\\\\

\item[Implication] 
$\begin{array}{c}  \Gamma\vdash t: A\rightarrow B\ \ \ \Gamma\vdash u:A \\ \hline
\Gamma\vdash tu:B
\end{array}\ \ \ \ $
$\begin{array}{c}  \Gamma, x:A \vdash u: B\\ \hline \Gamma\vdash \lambda x\, u:
A\rightarrow B
\end{array}$\\\\
\item[Disjunction Introduction] 
$\begin{array}{c}  \Gamma \vdash u: A\\ \hline \Gamma\vdash \inj{0}{u}: A\vee B
\end{array}\ \ \ \ $
$\begin{array}{c}  \Gamma \vdash u: B\\ \hline \Gamma\vdash\inj{1}{u}: A\vee B
\end{array}$\\\\

\item[Disjunction Elimination] $\begin{array}{c} \Gamma\vdash u: A\vee B\ \ \ \Gamma, x: A \vdash w_1: C\ \ \ \Gamma, y:B\vdash w_2:
C\\ \hline \Gamma\vdash  u\, [x.w_{1}, y.w_{2}]: C
\end{array}$\\\\

\item[Universal Quantification] 
$\begin{array}{c} \Gamma \vdash u:\forall \alpha A\\ \hline  \Gamma\vdash um: A[m/\alpha]
\end{array} $
$\begin{array}{c}  \Gamma \vdash u: A\\ \hline \Gamma\vdash \lambda \alpha\, u:
\forall \alpha A
\end{array}$\\

where $m$ is any term of  the language $\mathcal{L}$ and $\alpha$ does not occur
free in any formula $B$ occurring in $\Gamma$.\\

\item[Existential Quantification] 
$\begin{array}{c}\Gamma\vdash  u: A[m/\alpha]\\ \hline \Gamma\vdash (
m,u):
\exists
\alpha A
\end{array}$ \ \ \ \
$\begin{array}{c} \Gamma\vdash u: \exists \alpha A\ \ \ \Gamma, x: A \vdash t:C\\
\hline
\Gamma\vdash u\, [(\alpha, x). t]: C
\end{array} $\\

where $\alpha$ is not free in $C$
nor in any formula $B$ occurring in $\Gamma$.\\



\end{description}
}

\caption{Term Assignment Rules for $\IL+\HMP$}\label{fig:system}

\begin{description}
 \item[Reduction Rules for $\IL$]
 \[(\lambda x. u)t\mapsto u[t/x]\]
 \[ (\lambda \alpha. u)m\mapsto u[m/\alpha]\]
  \[ \proj{i}{\pair{u_0}{u_1}}\mapsto u_i, \mbox{ for i=0,1}\]
 \[\case{\inj{i}{u}}{x_{1}.t_{1}}{x_{2}.t_{2}}\mapsto t_{i}[u/x_{i}], \mbox{ for i=0,1} \]
 \[\dest{\enc{m}{u}}{x}{v} \mapsto v[m/\alpha][u/x], \mbox{ for each term $m$ of $\mathcal{L}$} \]
\end{description}
\caption{Reduction Rules for $\IL$ + $\HMP$}\label{fig:red-ilmp}
\end{figure*}

As we shall see in \Cref{thm:construct-il-mp}, the reason why $\HMP$ cannot be appear in the head of a closed proof term having existential type is that its premise $\lnot\lnot\exists \alpha\, P$ is never classically valid, let alone provable in intuitionistic logic. \begin{proposition}
  \label{theorem:no-sigma-taut}
  Assume that the symbol $\rightarrow$ does not occur in the propositional formula $P$. Then $\lnot\lnot\exists \alpha\, P$ is not classically provable.\end{proposition}
\begin{proof}
  We provide a semantical argument. $\lnot\lnot\exists \alpha\, P$ is classically provable if and only if it is classically valid and thus if and only if $\exists \alpha\, P$ is classically valid. For every such a  formula, we shall exhibit a model falsifying it. Consider the model $\mathfrak{M}$ where every $n$-ary predicate is interpreted as the empty $n$-ary relation.  We show by induction on the complexity of the formula $P$ that $P^\mathfrak{M} = \bot$ for every assignment of individuals to the free variables of $P$, and therefore $(\exists \alpha \, P)^\mathfrak{M} = \bot$.
  \begin{itemize}
  \item If $P$ is atomic, then by definition of $\mathfrak{M}$, we have $P^\mathfrak{M} = \bot$ for every assignment of the variables.
  \item If $P=P_1 \land P_2$, then since by induction $P_1^\mathfrak{M} = \bot$, $(P_1\land P_2)^\mathfrak{M} = \bot$
  \item If $P=P_1 \lor P_2$, then since by induction $P_1^\mathfrak{M} = \bot$ and $P_2^\mathfrak{M} = \bot$, $(P_1\lor P_2)^\mathfrak{M} = \bot$
  \end{itemize}
\end{proof}
 In order to derive constructivity of $\IL+\HMP$, we shall just have to inspect the normal forms of proof terms. Our main argument, in particular, will use the following well-known syntactic characterization of the shape of proof terms.
\begin{proposition}[Head of a Proof Term]
  \label{theorem:head-form}
  Every proof-term of $\IL+\HMP$ is of the form 
\[\lambda z_1 \dots \lambda z_n.\, r u_1 \dots u_k \]
where 
\begin{itemize}
\item $r$ is either a variable or a constant or a term corresponding to an introduction rule: $\lambda x . t$, $\lambda \alpha . t$, $\pair{t_1}{t_2}$, $\inj{i}{t}$, $\enc{m}{t}$
\item $u_1, \dots u_k$ are either proof terms, first order terms, or one of the following expressions corresponding to elimination rules: $\proj{i}{}$, $\case{}{x.w_1}{y.w_2}$, $\dest{}{x}{t}$.
\end{itemize}
\end{proposition}
\begin{proof}
Standard. 
\end{proof}
We are now able to prove that $\IL+\HMP$ is constructive.

\begin{theorem}[Constructivity of $\IL+\HMP$]
\label{thm:construct-il-mp}\mbox{}
\begin{enumerate}
\item
If $\IL+\HMP \vdash t: \exists \alpha \, A$, and $t$ is in normal form, then $t=\enc{m}{u}$ and $\IL+\HMP \vdash u: A[m/\alpha]$.
\item If $\IL+\HMP \vdash t: A \lor B$ and $t$ is in normal form, then either $t=\inj{0}{u}$ and $\IL+\HMP \vdash u: A$ or $t=\inj{1}{u}$ and $\IL+\HMP\vdash u: B$.
\end{enumerate}
\end{theorem}

\begin{proof}\mbox{}
\begin{enumerate}
\item  By \Cref{theorem:head-form}, $t$ must be of the form $r u_1\dots u_k$. Let us consider the possible forms of $r$.
  \begin{itemize}
  \item Since $t$ is closed, $r$ cannot be a variable.
  \item We show that $r$ cannot be  $\MP{P}$. If $r$ were $\MP{P} : \neg \neg \exists x \, P \to \exists \alpha \, P$ for some $P$, then $\IL+\mathsf{MP} \vdash u_1 : \neg \neg \exists \alpha \, P$. Since $\IL+\HMP$ is contained in classical logic, we have that $\neg \neg \exists \alpha \, P$ is classically provable.
However we know from \Cref{theorem:no-sigma-taut} that this cannot be the case, which is a contradiction.
  \item We also show that $r$ cannot be  $\Hypz{\prop{\bot\rightarrow P}}$.  Indeed, if $r$ were $\Hypz{\prop{\bot\rightarrow P}}$ for some $P$, then $\IL+\mathsf{MP} \vdash u_1 : \bot$, which is a contradiction.
  \item The only possibility is thus that $r$ is one among $\lambda x . t$, $\lambda \alpha . t$, $\pair{t_1}{t_2}$, $\inj{i}{t}$, $\enc{m}{t}$. In this case, $k$ must be 0 as otherwise we would have a redex. This means that $t=r$ 
   and thus $t=\enc{m}{u}$ with $\IL+\HMP \vdash u : A(m)$.
  \end{itemize}
  \item The proof goes along the same lines of case 1.
  \end{enumerate}
\end{proof}
Finally, we prove that $\IL+\HMP$ is not powerful enough to express full Markov's principle $\marp$. Intuitively, the reason is that $\IL+\HMP$ is a constructive system and thus cannot be strong enough to interpret classical reasoning. This  would indeed be the case if $\IL+\HMP$ proved $\marp$, an axiom which complements very well negative translations.
\begin{proposition}
  $\IL+\HMP\nvdash \marp$.
\end{proposition}
\begin{proof}
Suppose for the sake of contradiction that $\IL+\HMP\vdash \marp$.
Consider any proof in classical first-order logic of a simply existential statement $\exists \alpha\,\prop{P}$. By the G\"odel-Gentzen negative translation (see \cite{Troelstra}), we can then obtain an intuitionistic proof of $\neg \neg \exists \alpha \, \prop{P}^{N}$, where $\prop{P}^{N}$ is the negative translation of $\prop{P}$, and thus $\IL+\HMP\vdash\exists \alpha\, \prop{P}^{N}$. By \Cref{thm:construct-il-mp}, there is a first-order term $m$ such that $\IL+\HMP\vdash \prop{P}^{N}[m/\alpha]$. Since $\prop{P}^{N}[m/\alpha]$ is classically equivalent to  $\prop{P}[m/\alpha]$, we would have a single witness for every classically valid simply existential statement. But this is not possible: consider for example the first-order language $\mathcal{L}=\{P,a,b\}$ and the formula $F = (P(a) \lor P(b)) \to  P(\alpha)$ where $P$ is an atomic predicate. Then the formula $\exists \alpha \, F$ is  classically provable, but there is no term $m$ such that $F[m/\alpha]$ is valid, let alone provable:
\begin{itemize}
\item it cannot be $m=a$, as it is shown by picking a model where $P$ is interpreted as the set $\{a\}$
\item it cannot be $m=b$, because we can interpret $P$ as the set $\{b\}$.
\end{itemize}
\end{proof}

\section{Full Markov Principle and Restricted Excluded Middle in First-Order Logic}
\label{sec:ilemeno}

In this section we describe the natural deduction system and Curry-Howard correspondence  $\IL+\EM_{1}^{-}$, which arise by restricting the excluded-middle in classical natural deduction \cite{AschieriZH}. This computational system is based on delimited exceptions and a parallel operator. We will show that on one hand full Markov principle $\marp$ is provable in $\IL+\EM_{1}^{-}$ and, on the other hand, that $\IL+\marp$ derives all of the restricted classical reasoning that can be expressed in $\IL+\EM_{1}^{-}$, so that the two systems are actually equivalent. Finally, we show that the system $\IL+\EM_{1}^{-}$ is Herbrand constructive and that witnesses can effectively be computed.\\

All of the classical reasoning in $\IL+\EM_{1}^{-}$ is formally restricted to negative formulas.
\begin{definition}[Negative, Simply Universal Formulas]
  We denote propositional formulas as $\prop{P_1},\dots \prop{P_n},\prop{Q},\prop{R},\dots$. We say that a propositional formula is \emph{negative} whenever $\lor$ does not occur in it. Formulas of the form $\forall \alpha_{1}\ldots \forall \alpha_{n}\, \prop{P}$, with  $\prop{P}$ negative, will be called \emph{simply universal}.

\end{definition}

In order to computationally interpret Markov's principle, we consider the rule  $\EM_{1}^{-}$, which is obtained by restricting the conclusion of the excluded middle $\EM_{1}$ \cite{AschieriZH, ABB} to be a simply existential formula,  
\begin{prooftree}
  \AxiomC{$\Gamma, a: \forall \alpha \, \prop{P} \vdash u: \exists \beta \, \prop{Q}$}
  \AxiomC{$\Gamma, a:  \exists \alpha\, \neg \prop{P} \vdash v: \exists \beta \, \prop{Q}$}
  \RightLabel{$\emeno{1}$}
  \BinaryInfC{$\Gamma\vdash  \E{a}{u}{v} : \exists \beta \, \prop{Q}$}
\end{prooftree}
where both $\prop{P}$ and $\prop{Q}$ are negative formulas. This inference rule is complemented by the axioms:
\[\Gamma, a:{\forall \alpha \prop{P}}\vdash \Hyp{a}{\alpha}{\prop{P}}: \forall\alpha \prop{P}\]
\[\Gamma, a:{\exists \alpha \lnot \prop{P}}\vdash \Wit{a}{\alpha}{\neg \prop{P}}: \exists\alpha \lnot \prop{P}\]
These  last two rules correspond respectively to a term making an \emph{Hypothesis} and a term waiting for a \emph{Witness} and these terms are put in communication via $\emeno{1}$. A term of the form $\Hyp{a}{\alpha}{\prop{P}} m$, with $m$ first-order term, is said to be \emph{active}, if its only free variable is $a$: it represents a raise operator which has been turned on.   
The  term $\E{a}{u}{v}$ supports an exception mechanism: $u$ is the ordinary computation, $v$ is the exceptional one and $a$ is the communication channel. Raising exceptions is the task of the term $\Hyp{a}{\alpha}{\prop{P}}$, when it encounters a counterexample $m$ to $\forall\alpha\, \prop{P}$; catching exceptions is performed by the term $\Wit{a}{\alpha}{\neg \prop{P}}$.  In first-order logic, however, there is an issue: when should an exception be thrown? Since the truth of atomic predicates depends on models, one cannot know. Therefore, each time $\Hyp{a}{\alpha}{\prop{P}}$ is applied to a term $m$, a  new pair of parallel independent computational paths is created, according as to whether $\prop{P}[m/\alpha]$ is false or true. In one path the exception is thrown, in the other not, and the two computations  will never join again. To render this computational behaviour, we add the rule $\emme{0}$ of propositional excluded middle over negative formulas 
\begin{prooftree}
  \AxiomC{$\Gamma, a:  \neg \prop{P} \vdash u: A$}
  \AxiomC{$\Gamma, a:  \prop{P} \vdash v: A$}
  \RightLabel{$\emme{0}$}
  \BinaryInfC{$\Gamma\vdash  \Ez{}{u}{v} : A$}
\end{prooftree}
 even if in principle it is derivable from $\EM_{1}^{-}$; we also add the axiom
 \[\Gamma, a:\prop{P} \vdash \Hypz{\prop{P}}: \prop{P}\]
 We call the resulting system $\IL+\EM_1^-$ (\cref{fig:system-ilem}) and present its reduction rules in  \cref{fig:red}; they just form a restriction of the system $\IL+\EM$ described in \cite{AschieriZH}. The reduction rules are in \cref{fig:red} and are based on the following definition, which formalizes the raise and catch mechanism.

\begin{definition}[Exception Substitution]\label{def:witsub}
  \label{definition-witsub} Suppose $v$ is any proof term and $m$ is a term of $\mathcal{L}$. Then:
  \begin{enumerate}
  \item If every free occurrence of $a$ in $v$ is of the form $\Wit{a}{\alpha}{ \prop{P}}$, we define $$v[a:=m]$$
    as the term obtained from $v$ by replacing each subterm $\Wit{a}{\alpha}{\prop{P}}$ corresponding to a free occurrence of $a$ in $v$ by $(m, \Hypz{\prop{P}[m/\alpha]})$.

  \item If every free occurrence of $a$ in $v$ is of the form $\Hyp{a}{\alpha}{\prop{P}}$, we define $$v[a:=m]$$
    as the term obtained from $v$ by replacing  each subterm $\Hyp{a}{\alpha}{\prop{P}}m$ corresponding to a free occurrence of $a$ in $v$ by $\Hypz{\mathsf{P}[m/\alpha]}$.
  \end{enumerate}
\end{definition}

\begin{figure*}[!htb]
 
\footnotesize{


\begin{description}

\item[Grammar of Untyped Proof Terms]
\[t,u, v::=\ x\  |\ tu\ |\ tm\ |\ \lambda x\, u\  |\ \lambda \alpha\, u\ |\ \pair{t}{u}\ |\ \proj{0}{u}\ |\ \proj{1}{u} \ |\ \inj{0}{u}\ |\ \inj{1}{u}\  |\ \case{t}{x.u}{y.v}\ |\ (m,t)\ |\ t[(\alpha, x). u]\] \[|\ (\Ez{}{u}{v})\ |\ (\E{a}{u}{v})\ |\ \Hyp{a}{\alpha}{A}\ |\ \Wit{a}{\alpha}{\prop{P}}\ |\ \Hypz{\prop{P}} \]
where $m$ ranges over terms of $\mathcal{L}$, $x$ over proof-term variables, $\alpha$ over first-order variables,  $a$ over hypothesis variables, $A$ is either a negative formula  or simply universal formula with a negative propositional matrix, and $\prop{P}$ is negative.

We assume that in the term $\E{a}{u}{v}$ there is some formula $\prop{P}$, such that $a$ occurs free in $u$ only in subterms of the form $\Hyp{a}{\alpha}{\prop{P}}$ and $a$ occurs free in $v$ only in subterms of the form $\Wit{a}{\alpha}{\prop{P}}$, and the occurrences of the variables in $\prop{P}$ different from $\alpha$ are free in both $u$ and $v$.

\item[Contexts] With $\Gamma$ we denote contexts of the form $x_1:A_1, \ldots, x_n:A_n$, where each $x_{i}$ is a proof-term variable, and $x_{i}\neq x_{j}$ for $i\neq j$.

\item[Axioms] 
$\begin{array}{c}   \Gamma, x:{A}\vdash x: A
\end{array}\ \ \ \ $
$\begin{array}{c}   \Gamma, a:{\forall {\alpha}\, A}\vdash  \Hyp{a}{\alpha}{A}:  \forall{\alpha}\, A
\end{array}\ \ \ \ $
$\begin{array}{c}   \Gamma, a:{\exists \alpha\,  \prop{P}\vdash \Wit{a}{\alpha}{\prop{P}}:  \exists{\alpha}\,  \prop{P}}
\end{array}$\\
$\begin{array}{c}   \Gamma, a: \prop{P}\vdash  \Hypz{\prop{P}}:   \prop{P}
\end{array}\ \ \ \ $
$\begin{array}{c}   \Gamma \vdash  \Hypz{\prop{\bot\rightarrow P}}:   \bot\rightarrow \prop{P}
\end{array}\ \ \ \ $
\\

\item[Conjunction] 
$\begin{array}{c}  \Gamma \vdash u:  A\ \ \ \Gamma\vdash t: B\\ \hline \Gamma\vdash \pair{u}{t}:
A\wedge B
\end{array}\ \ \ \ $
$\begin{array}{c} \Gamma \vdash u: A\wedge B\\ \hline\Gamma \vdash \proj{0}{u}: A
\end{array}\ \ \ \ $
$\begin{array}{c}  \Gamma \vdash u: A\wedge B\\ \hline \Gamma\vdash \proj{1}{u} : B
\end{array}$\\\\

\item[Implication] 
$\begin{array}{c}  \Gamma\vdash t: A\rightarrow B\ \ \ \Gamma\vdash u:A \\ \hline
\Gamma\vdash tu:B
\end{array}\ \ \ \ $
$\begin{array}{c}  \Gamma, x:A \vdash u: B\\ \hline \Gamma\vdash \lambda x\, u:
A\rightarrow B
\end{array}$\\\\
\item[Disjunction Introduction] 
$\begin{array}{c}  \Gamma \vdash u: A\\ \hline \Gamma\vdash \inj{0}{u}: A\vee B
\end{array}\ \ \ \ $
$\begin{array}{c}  \Gamma \vdash u: B\\ \hline \Gamma\vdash\inj{1}{u}: A\vee B
\end{array}$\\\\

\item[Disjunction Elimination] $\begin{array}{c} \Gamma\vdash u: A\vee B\ \ \ \Gamma, x: A \vdash w_1: C\ \ \ \Gamma, y:B\vdash w_2:
C\\ \hline \Gamma\vdash  u\, [x.w_{1}, y.w_{2}]: C
\end{array}$\\\\

\item[Universal Quantification] 
$\begin{array}{c} \Gamma \vdash u:\forall \alpha A\\ \hline  \Gamma\vdash um: A[m/\alpha]
\end{array} $
$\begin{array}{c}  \Gamma \vdash u: A\\ \hline \Gamma\vdash \lambda \alpha\, u:
\forall \alpha A
\end{array}$\\

where $m$ is any term of  the language $\mathcal{L}$ and $\alpha$ does not occur
free in any formula $B$ occurring in $\Gamma$.\\

\item[Existential Quantification] 
$\begin{array}{c}\Gamma\vdash  u: A[m/\alpha]\\ \hline \Gamma\vdash (
m,u):
\exists
\alpha A
\end{array}$ \ \ \ \
$\begin{array}{c} \Gamma\vdash u: \exists \alpha A\ \ \ \Gamma, x: A \vdash t:C\\
\hline
\Gamma\vdash u\, [(\alpha, x). t]: C
\end{array} $\\

where $\alpha$ is not free in $C$
nor in any formula $B$ occurring in $\Gamma$.\\

\item[$\emme{0}$]$\begin{array}{c} \Gamma, a:  \neg \prop{P} \vdash u: C\ \ \ \ \Gamma, a:  \prop{P} \vdash v: C\\ \hline \Gamma\vdash  \Ez{}{u}{v} : C \end{array}\ (\prop{P}\text{ propositional and negative})$ 

\item[$\EM_1^-$]$\begin{array}{c} \Gamma, a: \forall \alpha \, \prop{P} \vdash u: \exists \beta \ \prop{Q} \ \ \ \ \Gamma, a:  \exists \alpha\, \neg \prop{P} \vdash v: \exists \beta \ \prop{Q} \\ \hline \Gamma\vdash  \E{a}{u}{v} : \exists \beta \ \prop{Q} \end{array} \ \text{(} \prop{P},\prop{Q} \text{ propositional and negative) }$

\end{description}
}


\caption{Term Assignment Rules for $\IL+\emeno{1}$}\label{fig:system-ilem}
\end{figure*}


\begin{figure*}[!htb]
\begin{description}
 \item[Reduction Rules for $\IL$]
 \[(\lambda x. u)t\mapsto u[t/x]\qquad (\lambda \alpha. u)m\mapsto u[m/\alpha]\]
  \[ \proj{i}{\pair{u_0}{u_1}}\mapsto u_i, \mbox{ for i=0,1}\]
 \[\case{\inj{i}{u}}{x_{1}.t_{1}}{x_{2}.t_{2}}\mapsto t_{i}[u/x_{i}], \mbox{ for i=0,1} \]
 \[\dest{\enc{m}{u}}{x}{v} \mapsto v[m/\alpha][u/x], \mbox{ for each term $m$ of $\mathcal{L}$} \]

 \item[Permutation Rules for $\emme{0}$]
 \[(\Ez{}{u}{v}) w \mapsto \Ez{}{uw}{vw} \]
 \[\proj{i}{(\Ez{}{u}{v})}  \mapsto \Ez{}{\proj{i}{u}}{\proj{i}{v}} \]
 \[\case{(\Ez{}{u}{v})}{x.w_{1}}{y.w_{2}} \mapsto \Ez{}{\case{u}{x.w_{1}}{y}{w_{2}}}{\case{v}{x.w_{1}}{y}{w_{2}}}\]
 \[\dest{(\Ez{}{u}{v})}{x}{w} \mapsto \Ez{}{\dest{u}{x}{w}}{\dest{v}{x}{w}}\]

 \item[Reduction Rules for $\emeno{1}$]
 \[\E{a}{u}{v}\mapsto u,\, \mbox{ if $a$ does not occur free in $u$ }\]
 \[\E{a}{u}{v}\mapsto \Ez{}{v [a:=m]}{(\E{a}{u [a:=m]}{v})},\mbox{ whenever $u$ has some \emph{active} subterm $\Hyp{a}{\alpha}{\prop{P}} m$}\]

\end{description}
\caption{Reduction Rules for $\IL$ + $\emeno{1}$}\label{fig:red}
\end{figure*}
As we anticipated, our system is capable of proving the  full Markov Principle  $\marp$ and thus its particular case $\HMP$.

\begin{proposition}[Derivability of $\marp$] \mbox{}
  $\IL + \EM_1^- \vdash \marp$
\end{proposition}
\begin{proof}
First note that with the use of $\emme{0}$,  we obtain that $\IL+\EM_{1}^{-}\vdash P \lor \neg P$ for any atomic formula $P$. Therefore $\IL+\EM_{1}^{-}$ can prove any propositional tautology, and in particular  $\IL+\EM_{1}^{-}\vdash\prop{P} \lor \prop{Q} \leftrightarrow \neg (\neg \prop{P} \land \neg \prop{Q})$ for any propositional formulas $\prop{P}, \prop{Q}$, thus proving that each propositional formula is equivalent to a negative one.

Consider now any instance $\neg \neg \exists \alpha\, \prop{Q} \to \exists \alpha\, \prop{Q}$ of $\marp$. Thanks to the previous observation, we obtain
$$\IL+\EM_{1}^{-}\vdash \big(\neg \neg \exists \alpha\, \prop{Q} \to \exists \alpha\, \prop{Q}\big) \leftrightarrow  \big(\neg \neg \exists \alpha\, \prop{P} \to \exists \alpha\, \prop{P}\big)$$
 for some negative formula $\prop{P}$ logically equivalent to $\prop{Q}$. The following formal proof shows that $\IL+\emeno{1} \vdash\neg \neg \exists \alpha\, \prop{P} \to \exists \alpha\, \prop{P}$.
\begin{footnotesize}
    \begin{prooftree}
      \AxiomC{$[\neg \neg \exists \alpha \, \prop{P}]_{(2)}$}

      \AxiomC{$[\exists \alpha \, \prop{P}]_{(1)}$}

      \AxiomC{$[\forall \alpha \, \neg \prop{P}]_{\emeno{1}}$}
      \UnaryInfC{$\neg \prop{P}$} \AxiomC{$[P]_{\exists}$}
      \BinaryInfC{$\bot$}
  
      \RightLabel{$\exists$} \BinaryInfC{$\bot$} \RightLabel{$_{(1)}$}
      \UnaryInfC{$\neg \exists \alpha \, \prop{P}$} \BinaryInfC{$\bot$}

      \UnaryInfC{$\exists \alpha \, \prop{P}$}

      \AxiomC{$[\exists \alpha \, \neg \neg \prop{P}]_{\emeno{1}}$}

      \AxiomC{$[\prop{P}]_{\emme{0}}$}

      \insertBetweenHyps{\hskip 0pt}

      \AxiomC{$[\neg \neg \, \prop{P}]_{\exists}$}
      \AxiomC{$[\neg \prop{P}]_{\emme{0}}$}

      \BinaryInfC{$\bot$} \UnaryInfC{$\prop{P}$}
      \insertBetweenHyps{\hskip 0pt} 
      \RightLabel{$\emme{0}$} \BinaryInfC{$\prop{P}$}
      \RightLabel{$\exists$} \BinaryInfC{$\prop{P}$}
      \UnaryInfC{$\exists \alpha \, \prop{P}$} \insertBetweenHyps{\hskip
        -20pt} \RightLabel{$\emeno{1}$}
      \BinaryInfC{$\exists \alpha \, \prop{P}$}
      \RightLabel{$_{(2)}$}
      \UnaryInfC{$\neg\neg\exists \alpha \, \prop{P} \to \exists \alpha \,
        \prop{P}$}
    \end{prooftree}
  \end{footnotesize}
  Finally, this implies $\IL+\emeno{1} \vdash \neg \neg \exists \alpha\, \prop{Q} \to \exists \alpha\, \prop{Q}$.
\end{proof}

Conversely, everything which is provable within our system can be proven by means of first-order logic with full Markov principle.

\begin{theorem}
  If  $\IL+\emeno{1} \vdash F$, then $\IL + \marp \vdash F$.
\end{theorem}

\begin{proof}

We just need to show that $\IL + \marp$ can prove  the rules $\emeno{1}$ and $\emme{0}$. For the case of $\emme{0}$, note that $\IL+\marp\vdash\neg \neg \prop{P} \to \prop{P}$ for all propositional formulas $\prop{P}$, thanks to  $\marp$. Since for every propositional $\prop{Q}$ we have $\IL+\marp\vdash \lnot\lnot (\prop{Q}\lor \lnot \prop{Q})$, we obtain $\IL+\marp\vdash\prop{Q} \lor \neg \prop{Q}$, and therefore $\IL+\marp$ can prove $\emme{0}$ by mean of an ordinary disjunction elimination.

In the case of $\emeno{1}$, if we are given the proofs of
\AxiomC{$\forall \alpha\, \prop{P}$}
\noLine
\UnaryInfC{\vdots}
\noLine
\UnaryInfC{$\exists \alpha\, \prop{C}$}
\DisplayProof
and
\AxiomC{$\exists \alpha \neg \prop{P}$}
\noLine
\UnaryInfC{\vdots}
\noLine
\UnaryInfC{$\exists \alpha \prop{C}$}
\DisplayProof
in $\IL+\marp$,
 the following derivation shows a proof of $\exists \alpha\, \prop{C}$ in $\IL+\marp$.
\begin{prooftree}
  \AxiomC{$[\forall \alpha \prop{P}]_{(1)}$} \noLine
  \UnaryInfC{\vdots} \noLine \UnaryInfC{$\exists \alpha \prop{C}$}

  \AxiomC{$[\neg \exists \alpha \prop{C}]_{(4)}$} \BinaryInfC{$\bot$}
  \RightLabel{(1)} \UnaryInfC{$\neg \forall \alpha \prop{P}$}

  \AxiomC{$[\exists \alpha \neg \prop{P}]_{(2)}$} \noLine
  \UnaryInfC{\vdots} \noLine \UnaryInfC{$\exists \alpha \prop{C}$}

  \AxiomC{$[\neg \exists \alpha \prop{C}]_{(4)}$} \BinaryInfC{$\bot$}
  \RightLabel{$_{(2)}$} \UnaryInfC{$\neg\exists \alpha \neg \prop{P}$}

  \AxiomC{$[\neg \prop{P}]_{(3)}$}
  \UnaryInfC{$\exists \alpha \, \neg \prop{P}$} \BinaryInfC{$\bot$}
  \RightLabel{$_{(3)}$} \UnaryInfC{$\neg \neg \prop{P}$}

  \AxiomC{$\neg \neg \, \prop{P} \to \, \prop{P}$}

  \BinaryInfC{$\prop{P}$}
  \UnaryInfC{$\forall \alpha \, \prop{P}$}

  \insertBetweenHyps{\hskip -20pt}

  \BinaryInfC{$\bot$} \RightLabel{(4)}
  \UnaryInfC{$\neg \neg \exists \alpha \prop{C}$}
  \AxiomC{$\neg\neg\exists \alpha \, \prop{C} \to \exists \alpha \, \prop{C}$}
  \insertBetweenHyps{\hskip -20pt}
  \BinaryInfC{$\exists \alpha \, \prop{C}$}
\end{prooftree}
\end{proof}
As in \cite{AschieriZH}, all of our main results about witness extraction are valid not only for closed terms, but also for quasi-closed ones, which are those containing only pure universal assumptions.

\begin{definition}[Quasi-Closed terms]\label{def:quasi}

 An untyped proof term $t$ is said to be \emph{quasi-closed}, if it  contains as free variables only  hypothesis variables $a_{1}, \ldots, a_{n}$,  such that each occurrence of them is of the form $\Hyp{a_{i}}{\vec{\alpha}}{\prop{P}_i}$, where  $\forall \vec{\alpha}\, \prop{P}_{i}$ is simply universal.
\end{definition}
$\IL+\emeno{1}$ with the reduction rules in figure \cref{fig:red} enjoys the Subject Reduction Theorem, as a particular case of the Subject Reduction for $\IL+\EM$ presented in \cite{AschieriZH}.
\begin{theorem}[Subject Reduction]\label{subjectred}
If $\Gamma \vdash t : C$ and $t \mapsto u$, then $\Gamma \vdash u : C$.
\end{theorem}
No term of $\IL+\emeno{1}$ gives rise to an infinite reduction sequence \cite{AschieriZH}.

\begin{theorem}[Strong Normalization]
Every term typable in $\IL+\emeno{1}$ is strongly normalizing.
\end{theorem}

We now update the characterization of proof-terms heads given in  \Cref{theorem:head-form} to the case of $\IL+\emeno{1}$.
\begin{theorem}[Head of a Proof Term]
  \label{theorem:head-form-em}
  Every proof term of $\IL+\emeno{1}$ is of the form:
\[\lambda z_1 \dots \lambda z_n . r u_1 \dots u_k \]
where 
\begin{itemize}
\item $r$ is either a variable $x$, a constant $\Hypz{P}$ or $\Hyp{a}{\alpha}{A}$ or $\Wit{a}{\alpha}{\prop{P}}$ or an excluded middle term $\E{a}{u}{v}$ or $\Ez{}{u}{v}$, or a term corresponding to an introduction rule $\lambda x . t$, $\lambda \alpha . t$, $\pair{t_1}{t_2}$, $\inj{i}{t}$, $\enc{m}{t}$
\item $u_1, \dots u_k$ are either lambda terms, first order terms, or one of the following expressions corresponding to elimination rules: $\proj{i}{}$, $\case{}{x.w_1}{y.w_2}$, $\dest{}{x}{t}$
\end{itemize}
\begin{proof}
Standard.
\end{proof}
\end{theorem}
We now study the shape of the normal terms with the most simple types. 
\begin{proposition}[Normal Form Property]\label{prop:pnf}
Let $\prop{P},\prop{P}_1,\dots \prop{P}_n$ be negative propositional formulas, $A_1, \dots, A_m$ simply universal formulas. Suppose that 
\[ \Gamma =  z_1: \prop{P}_1, \dots z_n: \prop{P}_n, a_1 : \forall \alpha_1 A_1, \dots a_m : \forall \alpha_m A_m \]
and $\Gamma \vdash t:\exists {\alpha}\, \prop{P}$ or $\Gamma \vdash t: \prop{P}$, with $t$ in normal form and having all its free variables among $z_1, \dots z_n, a_1, \dots a_m $. Then:
\begin{enumerate}

\item Every occurrence in $t$ of every  term  $\Hyp{a_{i}}{\alpha_{i}}{A_i}$ is of the active form $\Hyp{a_{i}}{\alpha_{i}}{A_i}m$, where $m$ is a  term of $\mathcal{L}$.

\item $t$ cannot be of the form $u\parallel_{a} v$.
 \end{enumerate} 
\end{proposition}

\begin{proof} 
We prove 1. and 2. simultaneously and by induction on  $t$. There are several cases, according to the shape of $t$:\\
\begin{itemize}
\item $t=(m, u)$, $\Gamma\vdash t:\exists {\alpha}\, \prop{P}$ and $\Gamma \vdash u: \prop{P}[m/\alpha]$. We immediately get 1. by induction hypothesis applied to $u$, while  2. is obviously verified.

\item $t=\lambda x\, u$, $\Gamma \vdash t: \prop{P}=\prop{Q}\rightarrow \prop{R}$ and $\Gamma, x: \prop{Q} \vdash u: \prop{R}$. We immediately get 1. by induction hypothesis applied to $u$, while  2. is obviously verified.

\item $t=\langle u, v\rangle$, $\Gamma \vdash t: \prop{P}=\prop{Q}\land \prop{R}$, $\Gamma \vdash u: \prop{Q}$ and $\Gamma \vdash v: \prop{R}$. We immediately get 1. by induction hypothesis applied to $u$, while  2. is obviously verified.

\item $t=\Ez{}{u}{v}$, $\Gamma, a: \neg \prop{Q} \vdash u: \exists {\alpha}\, \prop{P}$ (resp. $u: \prop{P}$) and $\Gamma, a: \prop{Q} \vdash v: \exists {\alpha}\,\prop{P}$ (resp. $v: \prop{P}$). We immediately get the thesis by induction hypothesis applied to $u$ and $v$, while 2. is obviously verified.

\item  $t=\E{a}{u}{v}$. We show that this is not possible. Note that $a$ must occur free in $u$, otherwise $t$ is not in normal form. Since  $\Gamma, a: \forall \beta\, A \vdash u: \exists {\alpha}\, \prop{P}$, 
we can apply the induction hypothesis to $u$, and obtain that all occurrences of hypothetical terms must be active; in particular, this must be the case for the occurrences of $\Hyp{a}{\beta}{A}$, but this is not possible since $t$ is in normal form.

\item $t=\Hyp{a_{i}}{\alpha}{A_{i}}$. This case is not possible, for  $\Gamma\vdash t:\exists {\alpha}\, \prop{P}$ or $\Gamma\vdash t: \prop{P}$.

\item $t=\Hypz{\prop{P}}$. In this case, 1. and 2. are trivially true.

\item $t$ is obtained by an elimination rule and by \Cref{theorem:head-form-em} we can write it as $r\, t_{1}\,t_{2}\ldots t_n$. Notice that in this case $r$ cannot correspond to an introduction rule neither be a term of the form $\E{a}{u}{v}$, because of the induction hypothesis, nor $\Ez{}{u}{v}$, because of the permutation rules and  $t$ being in normal form; moreover, $r$ cannot be $\Wit{b}{\alpha}{P}$, otherwise $b$ would be free in $t$ and $b\neq a_{1}, \ldots, a_{n}$.
We have now two remaining cases:
\begin{enumerate}
\item
$r=x_{i}$ (resp. $r=\Hypz{\prop{P}}$). Then, since $\Gamma \vdash x_{i}: \prop{P}_{i}$ (resp. $\Gamma \vdash \Hypz{\prop{P}}: \prop{P}$), we have that for each $i$, either $t_{i}$  is $\pi_{j}$ or $\Gamma\vdash t_{i}: \prop{Q}$, where $\prop{Q}$ is a negative propositional formula. By induction hypothesis, each $t_{i}$ satisfies 1. and also $t$. 2. is obviously verified.

\item  $r=\Hyp{a_{i}}{\alpha_i}{{A}_{i}}$. Then, $t_{1}$ is $m$, for some closed term of $\mathcal{L}$. Let $A_{i}=\forall\gamma_1\ldots\forall\gamma_l \, \prop{Q}$, with $\prop{Q}$ propositional, we have that for each $i$, either $t_i$ is a closed term $m_i$ of $\mathcal{L}$ or  $t_{i}$  is $\pi_{j}$ or $\Gamma\vdash t_{i}: \prop{R}$, where $\prop{R}$ is a negative propositional formula. By induction hypothesis, each $t_{i}$ satisfies 1. and thus also $t$, while 2. is obviously verified.
\end{enumerate}
\end{itemize}
\end{proof}

If we omit the parentheses, we will show that every normal proof-term having as type an existential formula can be written as $\Ez{}{\Ez{}{\Ez{}{v_0}{v_{1}}}{}\ldots}{v_{n}}$, where each $v_{i}$ is not of the form $\Ez{}{u}{v}$; if for every $i$, $v_i$ is of the form $(m_i,u_i)$, then we call the whole term an \emph{Herbrand normal form}, because it is essentially a list of the witnesses  appearing in an Herbrand disjunction. Formally:
\begin{definition}[Herbrand Normal Forms] \label{definition-hnf}
We define by induction a set of proof terms, called \emph{Herbrand normal forms}, as follows:
\begin{itemize}
\item Every normal proof-term $(m,u)$ is an Herbrand normal form;
\item if $u$ and $v$ are Herbrand normal forms, $\Ez{}{u}{v}$ is an Herbrand normal form.
\end{itemize}
\end{definition}

Our last task is to prove that all quasi-closed proofs of any existential statement $\exists \alpha\, A$ include an exhaustive sequence $m_{1}, m_{2}, \ldots, m_{k}$ of possible witnesses. This theorem is stronger than the usual Herbrand theorem for classical logic \cite{AschieriZH}, since we are stating it for any existential formula and not just for formulas with a single and existential quantifier.

\begin{theorem}[Herbrand Disjunction Extraction]\label{theorem-extraction}
Let $\exists\alpha\,A$ be any closed formula. Suppose $\Gamma \vdash t:  \exists \alpha\, A$ in $\IL+\emeno{1}$ for a quasi closed term $t$, and $t\mapsto^{*} t'$ with $t^\prime$ in normal form. Then $\Gamma \vdash t': \exists \alpha\, A$ and $t'$ is an Herbrand normal form
$$\Ez{}{\Ez{}{\Ez{}{(m_{0}, u_{0})}{(m_{1}, u_{1})}}{}\ldots}{(m_{k}, u_{k})}$$
Moreover, $\Gamma \vdash A[m_{1}/\alpha]\lor \dots \lor A[m_{k}/\alpha]$.
\end{theorem}

\begin{proof}
By the Subject Reduction Theorem \ref{subjectred}, $\Gamma\vdash t': \exists \alpha\, A$.
We proceed by induction on the structure of $t^\prime$. According to \Cref{theorem:head-form-em}, we can write $t^\prime$ as $r u_1\dots u_n$. Note that since $t'$ is quasi closed, $r$ cannot be a variable $x$; moreover, $r$ cannot be a term $\Hypz{\prop{P}}$ or $\Hyp{b}{\alpha}{B}$, otherwise $t'$ would not have type $\exists \alpha\, A$, nor a term $\Wit{b}{\alpha}{\prop{P}}$, otherwise $t'$ would not be quasi closed.  $r$ also cannot be of the shape $\E{a}{u}{v}$,  otherwise  $\Gamma \vdash \E{a}{u}{v} : \exists \alpha\, \prop{Q}$, for some negative propositional $\prop{Q}$, but from \Cref{prop:pnf} we know that this is not possible. By \Cref{theorem:head-form-em}, we are now left with only two possibilities.

\begin{enumerate}
\item $r$ is obtained by an introduction rule. Then $n=0$, otherwise there is a redex, and thus the only possibility is $t^\prime = r = \enc{n}{u}$ which is an Herbrand Normal Form.
\item $r = \Ez{}{u}{v}$. Again $n=0$, otherwise we could apply a permutation rule; then $t^\prime = r = \Ez{}{u}{v}$, and the thesis follows by applying the induction hypothesis on $u$ and $v$.
\end{enumerate}
We have thus shown that $t^\prime$ is an Herbrand normal form $$\Ez{}{\Ez{}{\Ez{}{(m_{0}, u_{0})}{(m_{1}, u_{1})}}{}\ldots}{(m_{k}, u_{k})}$$ 
 Finally, we have that for each $i$, $\Gamma_{i}\vdash u_{i}: A[m_{i}/\alpha]$, for the very same  $\Gamma_{i}$ that types $(m_{i}, u_{i})$ of type $\exists \alpha\, A$ in $t'$. Therefore, for each $i$, $\Gamma_{i}\vdash u_{i}^{+}: A[m_{1}/\alpha]\lor \dots \lor A[m_{k}/\alpha]$, where $u_{i}^{+}$ is of the form $\inj{i_{1}}{\ldots \inj{i_{k}}{u_{i}}\ldots }$.  We conclude that 
 $$\Gamma \vdash \Ez{}{\Ez{}{\Ez{}{u_{0}^{+}}{u_{1}^{+}}}{}\ldots}{u_{k}^{+}}: A[m_{1}/\alpha]\lor \dots \lor A[m_{k}/\alpha]$$
  \end{proof}

\section{Markov's Principle in Arithmetic}
\label{sec:arithmetic}
The original statement of Markov's principle refers to Arithmetic and can be formulated in the system of Heyting Arithmetic $\HA$ as
$$ \neg \neg \exists \alpha \, \prop{P} \to \exists \alpha \, \prop{P} \text{, for } \prop{P} \text{ atomic}$$
By adapting $\IL+\emeno{1}$ to Arithmetic, following \cite{ABB}, we will now provide a new computational interpretation of Markov's principle. Note first of all that propositional formulas are decidable in intuitionistic Arithmetic $\HA$: therefore we will not need the rule $\emeno{0}$ and the parallelism operator.  For the very same reason, we can expect the system $\HA+\emeno{1}$ to be constructive and the proof to be similar to the one of Herbrand constructivity for $\IL+\emeno{1}$.  In this section indeed we will give such a syntactic proof. We could also have used the realizability interpretation for $\HA+\emme{1}$ introduced in \cite{ABB} (see \cite{Manighetti}).

\subsection{The system $\HA+\emeno{1}$}

We will now introduce the system $\HA+\emeno{1}$. We start by defining the language:
\begin{definition}[Language of $\HA + \emeno{1}$]\label{definition-languagear}
  The language $\mathcal{L}$ of $\HA + \EM_1$ is defined as follows.
  \begin{enumerate}
  \item The terms of $\mathcal{L}$ are inductively defined as either variables $\alpha, \beta,\ldots$ or $0$ or $\suc(t)$ with $t\in\mathcal{L}$. A numeral is a term of the form $\suc\ldots \suc 0$. \\
  \item There is one symbol $\mathcal{P}$ for every primitive recursive relation over $\mathbb{N}$; with $\mathcal{P}^{\bot}$ we denote the symbol for the complement of the relation denoted by $\mathcal{P}$. The atomic formulas of $\mathcal{L}$ are all the expressions of the form $\mathcal{P}(t_{1}, \ldots, t_{n})$ such that $t_{1}, \ldots, t_{n}$ are terms of $\mathcal{L}$ and $n$ is the arity of $\mathcal{P}$. Atomic formulas will also be denoted as $\prop{P}, \mathsf{Q}, \prop{P_i}, \ldots$ and $\mathcal{P}(t_{1}, \ldots, t_{n})^{\bot}:=\mathcal{P}^{\bot}(t_{1}, \ldots, t_{n})$. \\
\item The formulas of $\mathcal{L}$ are built from atomic formulas of $\mathcal{L}$ by the connectives $\lor,\land,\rightarrow, \forall,\exists$ as usual, with quantifiers ranging over numeric variables $\alpha^{\Nat}, \beta^{\Nat}, \ldots$.\\
\end{enumerate}
\end{definition}

The system $\HA+\emeno{1}$ in \cref{fig:haemeno} extends the usual Curry-Howard correspondence for $\HA$ with our rule $\emeno{1}$ and is a restriction of the system introduced in \cite{ABB}. The purely universal arithmetical axioms are introduced by means of Post rules, as in Prawitz \cite{Prawitz}.

As we anticipated, there is no need for a parallelism operator. Therefore $\emeno{1}$ introduces a pure delimited exception mechanism, explained by the reduction rules in \cref{fig:F}: whenever we have a term $\E{a}{u}{v}$ and $\Hyp{a}{\alpha}{\prop{P}}m$ appears inside $u$, we can recursively \emph{check} whether $\prop{P}[m/\alpha]$ holds, and switch to the exceptional path if it doesn't; alternatively, if it does hold we can remove the instance of the assumption. When there are no free assumptions relative to $a$ left in $u$, we can forget about the exceptional path.

\begin{figure*}[!htb]
\begin{description}
\item[Grammar of Untyped Terms]
\[t,u, v::=\ x\  |\ tu\ |\ tm\ |\ \lambda x\, u\  |\ \lambda \alpha\, u\ |\ \pair{t}{u}\ |\ \proj{0}{u}\ |\ \proj{1}{u} \ |\ \inj{0}{u}\ |\ \inj{1}{u}\  |\ \case{t}{x.u}{y.v}\ |\ (m,t)\ |\ t[(\alpha, x). u]\] \[|\ (\E{a}{u}{v})\ |\ \Hyp{a}{\alpha}{\prop{P}}\ |\ \Wit{a}{\alpha}{\prop{P}}\ |\ \True \ |\ \rec u v m \ |\ \mathsf{r}t_{1}\ldots t_{n} \]

where $m$ ranges over terms of $\mathcal{L}$, $x$ over variables of the lambda calculus and $a$ over $\EM_1$ hypothesis variables. Moreover, in terms of the form $\E{a}{u}{v}$ there is a $\prop{P}$ such that all the free occurrences of $a$ in $u$ are of the form $\Hyp{a}{\alpha}{\prop{P}}$ and those in $v$ are of the form $\Wit{a}{\alpha}{\prop{P}^{\bot}}$.
\item[Contexts] With $\Gamma$ we denote contexts of the form $e_1:A_1, \ldots, e_n:A_n$, where $e_{i}$ is either a proof-term variable $x, y, z\ldots$ or a $\EM_{1}$ hypothesis variable $a, b, \ldots$

\item[Axioms] 
$\begin{array}{c} \Gamma, x:{A}\vdash x: A 
\end{array}$
$\begin{array}{c} \Gamma, a:{\forall \alpha^{\Nat} \prop{P}}\vdash \Hyp{a}{\alpha}{\prop{P}}: \forall\alpha^{\Nat} \prop{P}
\end{array}$
$\begin{array}{c} \Gamma, a:{\exists \alpha^{\Nat} \prop{P}^\bot}\vdash \Wit{a}{\alpha}{P}: \exists\alpha^{\Nat} \prop{P}^\bot
\end{array}$

\item[Conjunction] 
$\begin{array}{c}  \Gamma \vdash u:  A\ \ \ \Gamma\vdash t: B\\ \hline \Gamma\vdash \pair{u}{t}:
A\wedge B
\end{array}\ \ \ \ $
$\begin{array}{c} \Gamma \vdash u: A\wedge B\\ \hline\Gamma \vdash \proj{0}{u}: A
\end{array}\ \ \ \ $
$\begin{array}{c}  \Gamma \vdash u: A\wedge B\\ \hline \Gamma\vdash \proj{1}{u} : B
\end{array}$\\\\

\item[Implication] 
$\begin{array}{c}  \Gamma\vdash t: A\rightarrow B\ \ \ \Gamma\vdash u:A \\ \hline
\Gamma\vdash tu:B
\end{array}\ \ \ \ $
$\begin{array}{c}  \Gamma, x:A \vdash u: B\\ \hline \Gamma\vdash \lambda x\, u:
A\rightarrow B
\end{array}$\\\\
\item[Disjunction Introduction] 
$\begin{array}{c}  \Gamma \vdash u: A\\ \hline \Gamma\vdash \inj{0}{u}: A\vee B
\end{array}\ \ \ \ $
$\begin{array}{c}  \Gamma \vdash u: B\\ \hline \Gamma\vdash\inj{1}{u}: A\vee B
\end{array}$\\\\

\item[Disjunction Elimination] $\begin{array}{c} \Gamma\vdash u: A\vee B\ \ \ \Gamma, x: A \vdash w_1: C\ \ \ \Gamma, y:B\vdash w_2:
C\\ \hline \Gamma\vdash  u\, [x.w_{1}, y.w_{2}]: C
\end{array}$\\\\

\item[Universal Quantification] 
$\begin{array}{c} \Gamma \vdash u:\forall \alpha^\Nat A\\ \hline  \Gamma\vdash um: A[m/\alpha]
\end{array} $
$\begin{array}{c}  \Gamma \vdash u: A\\ \hline \Gamma\vdash \lambda \alpha\, u:
\forall \alpha^\Nat A
\end{array}$\\

where $m$ is any term of  the language $\mathcal{L}$ and $\alpha$ does not occur
free in any formula $B$ occurring in $\Gamma$.\\

\item[Existential Quantification] 
$\begin{array}{c}\Gamma\vdash  u: A[m/\alpha]\\ \hline \Gamma\vdash (
m,u): \exists \alpha^\Nat A \end{array}$ \ \ \ \
$\begin{array}{c} \Gamma\vdash u: \exists \alpha^\Nat A\ \ \ \Gamma, x: A \vdash t:C\\ \hline \Gamma\vdash u\, [(\alpha, x). t]: C \end{array} $\\

where $\alpha$ is not free in $C$
nor in any formula $B$ occurring in $\Gamma$.\\

\item[Induction] 
$\begin{array}{c} \Gamma\vdash u: A(0)\ \ \ \Gamma\vdash v:\forall \alpha^{\Nat}. A(\alpha)\rightarrow A(\suc(\alpha))\\ \hline \Gamma\vdash \rec uvm : A[m/\alpha] \end{array}$, where $m$ is a term of $\mathcal{L}$

\item[Post Rules] 
$\begin{array}{c} \Gamma\vdash u_1: A_1\ \Gamma\vdash u_2: A_2\ \cdots \ \Gamma\vdash u_n: A_n\\ \hline\Gamma\vdash u: A \end{array}$

where $A_1,A_2,\ldots,A_n,A$ are atomic formulas of $\HA$ and the rule is a Post rule for equality, for a Peano axiom or for a classical propositional
tautology or for booleans and if $n>0$, $u=\mathsf{r} u_{1}\ldots u_{n}$, otherwise $u=\True$.

\item[$\EM_1^-$]$\begin{array}{c} \Gamma, a: \forall \alpha \, \prop{P} \vdash u: \exists \beta \ \prop{Q} \ \ \ \ \Gamma, a:  \exists \alpha\, \neg \prop{P} \vdash v: \exists \beta \ \prop{Q} \\ \hline \Gamma\vdash  \E{a}{u}{v} : \exists \beta \ \prop{Q} \end{array} \ \text{(} \prop{P} \text{ atomic},\prop{Q} \text{ negative propositional) }$
\end{description}

\caption{Term Assignment Rules for $\HA+\EM_{1}$}
\label{fig:haemeno}
\end{figure*}

\begin{figure*}[!htb]
\begin{description}
\item[Reduction Rules for $\HA$] 
 \[(\lambda x. u)t\mapsto u[t/x]\qquad (\lambda \alpha. u)m\mapsto u[m/\alpha]\]
  \[ \proj{i}{\pair{u_0}{u_1}}\mapsto u_i, \mbox{ for i=0,1}\]
 \[\case{\inj{i}{u}}{x_{1}.t_{1}}{x_{2}.t_{2}}\mapsto t_{i}[u/x_{i}], \mbox{ for i=0,1} \]
 \[\dest{\enc{m}{u}}{x}{v} \mapsto v[m/\alpha][u/x], \mbox{ for each term $m$ of $\mathcal{L}$} \]
\[\rec u v 0 \mapsto u\]
\[\rec u v (\suc n) \mapsto v n (\rec u v n), \mbox{ for each numeral $n$} \]

 \item[Reduction Rules for $\emeno{1}$]
 \[\E{a}{u}{v}\mapsto u,\, \mbox{ if $a$ does not occur free in $u$ }\]
 \[\E{a}{u}{v}\mapsto v[a:=n],\mbox{ if $\Hyp{a}{\alpha}{\prop{P}}n$ occurs in $u$ and $\prop{P}[n/\alpha]$ is \emph{closed} and $\prop{P}[n/\alpha] = \False$}\]
\[(\Hyp{a}{\alpha}{\prop{P}})n \mapsto \True \mbox{ if $\prop{P}[n/\alpha]$ is \emph{closed} and $\prop{P}[n/\alpha] \equiv \True$}\] 

\end{description}
\caption{Reduction Rules for $\HA$ + $\EM_{1}$}
\label{fig:F}
\end{figure*}

Similarly to the previous sections, we extend the characterization of the proof-term heads to take into account the new constructs.

\begin{theorem}[Head of a Proof Term]
  \label{theorem:head-form-ha}
  Every proof term of $\HA+\emeno{1}$ is of the form:
\[\lambda z_1 \dots \lambda z_n . r u_1 \dots u_k \]
where 
\begin{itemize}
\item $r$ is either a variable $x$, a constant $\Hyp{a}{\alpha}{P}$, $\Wit{a}{\alpha}{P}$, $\mathsf{r}$ or $\rec$, an excluded middle term $\E{a}{u}{v}$, or a term corresponding to an introduction rule $\lambda x . t$, $\lambda \alpha . t$, $\pair{t_1}{t_2}$, $\inj{i}{t}$, $\enc{m}{t}$
\item $u_1, \dots u_k$ are either lambda terms, first order terms, or one of the following expressions corresponding to elimination rules: $\proj{i}{}$, $\case{}{x.w_1}{y.w_2}$, $\dest{}{x}{t}$
\end{itemize}

\end{theorem}

The new system proves exactly the same formulas that can be proven by making use of Markov's principle in Heyting Arithmetic.
\begin{theorem}
  For any formula $F$ in the language $\mathcal{L}$, $\HA+\marp \vdash F$ if and only if $\HA + \EM_1^- \vdash F$
\end{theorem}
\begin{proof}
  The proof is identical as the one in the previous section. 
\end{proof}
$\HA+\emeno{1}$ with the reduction rules in figure \cref{fig:red} enjoys the Subject Reduction Theorem \cite{ABB, Manighetti}.
\begin{theorem}[Subject Reduction]\label{subjectred3}
If $\Gamma \vdash t : C$ and $t \mapsto u$, then $\Gamma \vdash u : C$.
\end{theorem}
No term of $\HA+\emeno{1}$ gives rise to an infinite reduction sequence \cite{AschieriCOS}.

\begin{theorem}[Strong Normalization]
Every term typable in $\HA+\emeno{1}$ is strongly normalizing.
\end{theorem}

\subsection{$\HA+\emeno{1}$ is Constructive}

We can now proceed to prove the constructivity of the system, that is the disjunction and existential properties. We will do this again by inspecting the normal forms of the proof terms; the first thing to do is adapting \Cref{prop:pnf} to $\HA+\emeno{1}$.

\begin{proposition}[Normal Form Property]\label{prop:pnf2}
Let $\prop{P},\prop{P}_1,\dots \prop{P}_n$ be negative propositional formulas, $A_1, \dots A_m$ simply universal formulas. Suppose that 
\[ \Gamma =  z_1: \prop{P}_1, \dots z_n: \prop{P}_n, a_1 : \forall \alpha_1 A_1, \dots a_m : \forall \alpha_m A_m \]
and $\Gamma \vdash t:\exists {\alpha}\, \prop{P}$ or $\Gamma \vdash t: \prop{P}$, with $t$ in normal form and having all its free variables among $z_1, \dots z_n, a_1, \dots a_m $. Then:
\begin{enumerate}

\item Every occurrence in $t$ of every  term  $\Hyp{a_{i}}{\alpha_{i}}{A_i}$ is of the active form $\Hyp{a_{i}}{\alpha_{i}}{A_i}m$, where $m$ is a  term of $\mathcal{L}$

\item $t$ cannot be of the form $u\parallel_{a} v$.
 \end{enumerate} 
\end{proposition}

\begin{proof} 
The proof is identical to the proof of \Cref{prop:pnf}. We just need to consider the following additional cases:
\begin{itemize}
\item $t=\mathsf{r} t_1 t_2 \dots t_n$. Then $\Gamma \vdash t_i : \prop{Q}_i$ for some atomic $\prop{Q}_i$ and for $i=1 \dots n$; 1. holds by applying the inductive hypothesis to the $t_i$, while 2. is obviously verified.

\item $t = \rec t_1 \dots t_n$. This case is not possible, otherwise since $t_{3}$ is a numeral and thus $t$ would not be in normal form.
\end{itemize}
\end{proof}

Thanks to this, we can now state the main theorem. The proof of the existential property is the same as the one for \Cref{theorem-extraction}: we just need to observe that since we don't have a parallelism operator in $\HA+\emeno{1}$, every Herbrand disjunction will consist of a single term. The disjunction property will follow similarly.
\begin{theorem}[Constructivity of $\HA+\emeno{1}$] \mbox{}
  \begin{itemize}
  \item If $\HA+\emeno{1} \vdash t : \exists \alpha A$, then there exists a term $t' = (n,u)$ such that $t\mapsto^{*}t'$ and $\HA+\emeno{1} \vdash u : A[n/\alpha]$
  \item If $\HA+\emeno{1} \vdash t : A \lor B$, then there exists a term $t'$ such that $t\mapsto^{*}t'$ and either $t'=\inj{0}{u}$ and $\HA+\emeno{1} \vdash u : A$, or $t'=\inj{1}{u}$ and $\HA+\emeno{1} \vdash u: B$
  \end{itemize}
\end{theorem}
\begin{proof}\mbox{}
For both cases, we start by considering a term $t'$ such that $t \mapsto^* t'$ and $t'$ is in normal form. By the Subject Reduction \Cref{subjectred} we have that $\HA+\emeno{1} \vdash t' : \exists \alpha A$ (resp. $\HA+\emeno{1} \vdash t': A \lor B$). By \Cref{theorem:head-form-em} we can write $t'$ as $r t_1 \dots t_n$. Since $t'$ is closed, $r$ cannot be a variable $x$ or a term $\Hyp{a}{\alpha}{\prop{P}}$ or $\Wit{a}{\alpha}{\prop{P}}$; moreover it cannot be $\mathsf{r}$, otherwise the type of $t'$ would have to be atomic, and it cannot be $\rec$, otherwise the term would not be in normal form.  $r$ also cannot have been obtained by $\emeno{1}$, otherwise $\HA+\emeno{1} \vdash r : \exists \alpha\prop{P}$, for $\prop{P}$ atomic and $r=\E{a}{t_1}{t_2}$; but this is not possible due to \Cref{prop:pnf2}. Therefore, $r$ must be obtained by an introduction rule. We distinguish now the two cases:
  \begin{itemize}
  \item $\HA+\emeno{1} \vdash t' : \exists \alpha B$. Since the term is in normal form, $n$ has to be 0, that is $t'=r$ and  $r=\enc{n}{u}$; hence also $\HA+\emeno{1} \vdash u : A(n)$.
  \item $\HA+\emeno{1} \vdash t' : A \lor B$.  Then either $t'=\inj{0}{u}$, and so $\HA+\emeno{1} \vdash u : A$, or $t'=\inj{1}{u}$, and so $\HA+\emeno{1} \vdash u: B$.
  \end{itemize}
\end{proof}

\end{document}